\documentclass{svproc}      
\usepackage{amssymb}
\usepackage{amsmath}
\usepackage{graphicx}       
\usepackage{xcolor}
\usepackage{booktabs}
\usepackage{multirow}

\newenvironment{sistem}
{\left\lbrace\begin{array}{@{}l@{}}}
{\end{array}\right.}

\usepackage{url}

\begin{document}                                         
\mainmatter                                              

\title{Action potential dynamics on heterogenous neural networks: from kinetic to macroscopic equations}
\titlerunning{Kinetic models for action potential} 

\author{Marzia Bisi\inst{1} \and Martina Conte\inst{1} \and Maria Groppi$^*$\inst{1}}

\authorrunning{Marzia Bisi et al.}  

\tocauthor{Marzia Bisi, Martina Conte, Maria Groppi}         

\institute{$^1$Department of Mathematical, Physical, and Computer Science, University of Parma, Parco Area delle Scienze 53/A - 43124 Parma, Italy\\ \email{$^*$ Corresponding author: maria.groppi@unipr.it}}

\maketitle                                               

\begin{abstract}     
In the context of multi-agent systems of binary interacting particles, a kinetic model for action potential dynamics on a neural network is proposed, accounting for heterogeneity in the neuron-to-neuron connections, as well as in the brain structure. Two levels of description are coupled: in a single area, pairwise neuron interactions for the exchange of membrane potential are statistically described; among different areas, a graph description of the brain network topology is included. Equilibria of the kinetic and macroscopic settings are determined and numerical simulations of the system dynamics are performed with the aim of studying the influence of the network heterogeneities on the membrane potential propagation and synchronization.                                 
\end{abstract}                                           

\section{Introduction}\label{sec:intro}  
In recent years, the concepts of collective behaviors and network dynamics have gained a lot of attention in the context of complex systems. In particular, these have emerged as central features in the description of neuronal processes in the brain, both at the microscopic and macroscopic levels. It is well-recognized, in fact, that emerging collective behaviors, like neuron synchronization, are strongly influenced by the microscopic interactions of the individual agents involved in the process. 

Methodological approaches coming from the kinetic theory have shown great potential to tackle systems that, at the microscopic scale, may be described as large ensembles of interacting particles. It provides, in fact, a flexible framework integrating the stochastic particle level, where elementary cell dynamics take place, with the macroscopic aggregate level. Without pretending to be exhaustive, we mention that, in the recent literature, examples of the use of kinetic theory methods for biophysical applications are found e.g. in the modeling of the onset and propagation of Alzheimer's disease \cite{bertsch2017JPA,bertsch2017alzheimer} as well as of tumor growth, autoimmune diseases, and therapeutic control strategies \cite{preziosi2021JTB,conte2018DCDSB,conte2023M3AS,DellaMarca22}. Along with kinetic theory methods, network models play a key role in the context of interacting multi-agent systems, since they allow to describe particle interactions in relation to the background graph structure. Examples of kinetic descriptions of networked interactions can be found, for instance, in opinion formation problems \cite{burger2021VJM,toscani2018PRE} or in the modeling of the spread of infectious diseases \cite{loy2021MBE}. 

In this paper, we adopt the methodological framework proposed in \cite{conte2023kinetic} to describe action potential dynamics in heterogenous neural networks. We use the kinetic word {\it particle} to indicate the fundamental microscopic entity whose interactions shape the dynamics of the entire system, i.e., a brain cell. We consider the brain network structure introduced in \cite{sporns2005human}, assuming that  particles  populate the vertices of a relatively small graph that describes the topology of the connections among different particle communities (brain macro-areas). Within each node, due to the large number of particles and connections, the interactions are statistically described, starting from the modeling of the neuron-to-neuron exchanges of electrical stimuli. This combined approach allows to account for different types of heterogeneities, in both the graph structure and the neuron-to-neuron interactions. We derive the systems of coupled differential equations that characterize the network dynamics relevant to two choices of the neuron connection distribution function. Then, we investigate the complete synchronization of the resulting neural networks, as well as possible emerging asynchronous states. 

The paper is organized as follows. In Section \ref{Sec_Prelim}, we first recall the main elements of the network model proposed in \cite{conte2023kinetic} and, then, we specify the microscopic interaction rule and the neurons connection distribution and we derive the macroscopic systems under investigation. In Section \ref{Sec_Eq}, we analyze the equilibria of both the kinetic and the macroscopic settings. Section \ref{simulation} collects three sets of numerical tests aimed at showing the influence of the heterogenous network on the propagation and synchronization of the action potential. Finally, in Section \ref{conclusion} we draw some conclusions and illustrate possible outcomes of this work.

\section{Preliminaries}\label{Sec_Prelim}
Starting from the general setting proposed in \cite{conte2023kinetic}, in this study we focus on the propagation of the membrane potential generated by neurons and we analyze the impact of different heterogenous aspects, which may characterize the network, on its dynamics. The modeling framework is based on a combination of a statistical description of the neuron interactions within each brain macro-area \cite{bullmore2009complex,sporns2005human} with a classical description, based on graph theory, for the connection between the different areas. In the following, we first recall the main aspects of the considered modeling framework, and then we describe the interaction rule that describes the microscopic interactions.

\subsection{Model description}\label{Model_derivation}
Let $t\ge0$ be the time variable, to describe the membrane potential evolution we define $\mathcal{V}_t:=\mathcal{V}(t)$ the dimensionless {\it membrane potential} of a single neuron \cite{izhikevich2007dynamical}, which we assume to evolve according to a neuron-to-neuron transmission mechanism, and $\mathcal{W}_t:=\mathcal{W}(t)$ the associated {\it recovery variable}. The membrane potential is defined as the difference in electric potential between the interior and the exterior of a biological cell. Its typical dynamic is the so-called {\it action potential}, defined as a sudden, fast, transitory, and propagating change of the resting membrane potential. The recovery variable is introduced to act as negative feedback on the membrane potential. Then, we define $\mathcal{C}$ the {\it degree of connections} of a neuron and $\mathcal{X}$ its {\it spatial location} within the brain network. We assume that $\mathcal{C}$ and $\mathcal{X}$ do not depend on time, as the evolution of the action potential is much faster than the temporal window in which cerebral tissue changes may happen. Thus, the microscopic state of a single neuron is characterized by the quartet $(\mathcal{X},\mathcal{V},\mathcal{W},\mathcal{C})$ in the state space $\mathbb{I}\times[0, 1]\times[0,1]\times[0,\mathcal{C}_{\text{max}}]$, with ${\mathbb{I}:=\{1,...,N\}}$, $N$ being the number of brain macro-areas, and $\mathcal{C}_{\text{max}}$  a huge, but finite, upper bound for the number of connections of a single neuron.  As such, the microscopic variables $\mathcal{V}$ and $\mathcal{W}$ have to be regarded as independent of $\mathcal{C}$ and $\mathcal{X}$, in the sense that their evolution is not a function of $\mathcal{C}$ or $\mathcal{X}$ and vice versa.

Let us consider a generic pair of neurons with microscopic states $(\mathcal{X},\mathcal{V}_t,\mathcal{W}_t,\mathcal{C})$, ${(\mathcal{X}^*,\mathcal{V}^*_t,\mathcal{W}^*_t,\mathcal{C}^*) \in \mathbb{I}\times [0,1]\times[0,1]\times [0,\mathcal{C}_{\text{max}}]}$. Introducing a Bernoulli random variable $T_p\in\{0,1\}$ of parameter $p$ that describes whether an interaction leading to a change of microscopic state takes place ($\mathbb{T}_p=1$) or not ($\mathbb{T}_p=0$), we assume that in a short time interval $\Delta t>0$ there is a probability proportional to $\Delta t\,B(\mathcal{X},\mathcal{X}^*) G(\mathcal{C},\mathcal{C}^*)$ that a neuron undergoes the transmission mechanism. The collision kernel $B(\mathcal{X},\mathcal{X}^*)G(\mathcal{C},\mathcal{C}^*)$ takes into account the probability that a neuron with microscopic states $(\mathcal{X},\mathcal{V},\mathcal{W},\mathcal{C})$ is connected by a synapsis with a neuron with microscopic state $(\mathcal{X}^*,\mathcal{V}^*,\mathcal{W}^*,\mathcal{C}^*)$. Even though, in general, the collision kernel can be given as a unique function of $(\mathcal{X},\mathcal{C})$, here we assume that it can be factorized in such a way that $B(\mathcal{X},\mathcal{X}^*)$ describes the influence of the neuron location on the interactions, while $G(\mathcal{C},\mathcal{C}^*)$ refers to their degree of connections. The specific definition of $B(\mathcal{X},\mathcal{X}^*)$ is related to the brain network description and is provided in the following section. Concerning the choice of the function $G$, representing the probability of a neuron connection in terms of the number of connections $\mathcal{C}$ and $\mathcal{C}^*$ of the interacting neurons, we analyze two distinct cases:
\[
{G(\mathcal{C},\mathcal{C}^*)=1} \qquad \text{or} \qquad {G(\mathcal{C},\mathcal{C}^*)=\mathcal{C}^*}\,.
\]
The former indicates that all neurons within a single area are connected to each other, while the latter assumes that the greater the number of connections of the postsynaptic neuron, the greater the probability of undergoing the transmission. Such choice for $G(\mathcal{C},\mathcal{C}^*)$ is particularly relevant in the context of disease propagation, like Alzheimer. In fact, in this case, the disease propagates from ill to healthy neurons and the healthier a neuron (thus, the higher the number of its connections) is, the higher the probability of interacting with ill neurons is.  The microscopic law for the pair $(\mathcal{X},\mathcal{V}_t,\mathcal{W}_t,\mathcal{C})$, $(\mathcal{X}^*,\mathcal{V}^*_t,\mathcal{W}^*_t,\mathcal{C}^*)$ can be written as
\begin{equation}
\begin{sistem}
(\mathcal{V}_{t+\Delta t},\mathcal{W}_{t+\Delta t})=(1-\mathbb{T}_p)(\mathcal{V}_t,\mathcal{W}_t)+\mathbb{T}_p(\mathcal{V}'_t,\mathcal{W}'_t)\\[0.2cm]
(\mathcal{V}^*_{t+\Delta t},\mathcal{W}^*_{t+\Delta t})=(1-\mathbb{T}_p)(\mathcal{V}^{*}_t,\mathcal{W}^{*}_t)+\mathbb{T}_p(\mathcal{V}^{*'}_t,\mathcal{W}^{*'}_t)\,.\\
\end{sistem}
\end{equation}
where $(\mathcal{V}'_t,\mathcal{W}'_t):=(\mathcal{V}_t,\mathcal{W}_t)+L(\mathcal{X},\mathcal{X}^*,\mathcal{V}_t,\mathcal{W}_t,\mathcal{V}^*_t,\mathcal{W}^*_t)$. Here, the function
$$L(\mathcal{X},\mathcal{X}^*,\mathcal{V}_t,\mathcal{W}_t,\mathcal{V}^*_t,\mathcal{W}^*_t):\mathbb{I}\times\mathbb{I}\times[0, 1]^4\to[0, 1]^2$$ 
describes the binary interactions rule. 

To derive the Boltzmann-type equations for the evolution of $\mathcal{V}_t$ and $\mathcal{W}_t$ we start from the relation
\[
\phi(\mathcal{X},\mathcal{V}_{t+\Delta t},\mathcal{W}_{t+\Delta t},\mathcal{C})=\phi(\mathcal{X}, (\mathcal{V}_t,\mathcal{W}_t)+\mathbb{T}_pL(\mathcal{X},\mathcal{X}^*,\mathcal{V}_t,\mathcal{W}_t,\mathcal{V}^*_t,\mathcal{W}^*_t),\mathcal{C})
\]
where $\phi:\mathbb{I}\times [0,1]\times[0,1]\times[0,\mathcal{C}_{\text{max}}]\to \mathbb{R}$ is a test function representing any observable quantity that can be computed out of the microscopic state. Taking the average of both sides and computing the mean with respect to $\mathbb{T}_p$ (denoted with angular brackets), we get
\begin{equation}\label{ave_eq}
\begin{split}
&\langle\phi(\mathcal{X},\mathcal{V}_{t+\Delta t},\mathcal{W}_{t+\Delta t},\mathcal{C})\rangle=\langle\phi(\mathcal{X},\mathcal{V}_{t},\mathcal{W}_{t},\mathcal{C})\rangle\\[0.2cm]
&+\Delta t \, \langle B(\mathcal{X},\mathcal{X}^*)G(\mathcal{C},\mathcal{C}^*)\phi(\mathcal{X},(\mathcal{V}_t,\mathcal{W}_t)+L(\mathcal{X},\mathcal{X}^*,\mathcal{V}_t,\mathcal{W}_t,\mathcal{V}^*_t,\mathcal{W}^*_t),\mathcal{C})\rangle\\[0.2cm]
&-\Delta t \,\langle B(\mathcal{X},\mathcal{X}^*)G(\mathcal{C},\mathcal{C}^*)\phi(\mathcal{X},\mathcal{V}_{t},\mathcal{W}_{t},\mathcal{C})\rangle
\end{split}
\end{equation}
In the limit $\Delta t \to 0^+$, equation \eqref{ave_eq} reads
\begin{equation}\label{limit_eq}
\dfrac{d}{dt}\langle\phi(\mathcal{X},\mathcal{V}_{t},\mathcal{W}_{t},\mathcal{C})\rangle=\langle B(\mathcal{X},\mathcal{X}^*)G(\mathcal{C},\mathcal{C}^*)\left[\phi(\mathcal{X},\mathcal{V}_{t}',\mathcal{W}_{t}',\mathcal{C})- \phi(\mathcal{X},\mathcal{V}_{t},\mathcal{W}_{t},\mathcal{C})\right]\rangle\,.
\end{equation}

Let us identify each neuron by $(x,v,w,c)$, where $v$ is the membrane potential, $w$ the recovery variable, $c$ its degree of connections, and $x$ a discrete variable indicating the macro-area to which it belongs. Moreover, let us define $f(x,v,w,c,t):\mathbb{I}\times[0,1]\times[0,1]\times[0,c_{\text{max}}]\times\mathbb{R}_+\to\mathbb{R}$ the kinetic neuron distribution function such that
\begin{equation}\label{f_inhomo}
f=f(x,v,w,c,t):=\sum_{i=1}^Nf_i(v,w,c,t)\delta(x-i)\,,
\end{equation}
with $f_i(v,w,c,t)=f(\mathcal{X}=i,\mathcal{V}_t=v,\mathcal{W}_t=w,\mathcal{C}=c,t)$ being the kinetic distribution function in the $i$-th macro-areas. It holds
\[
\int\limits_0^1\int\limits_0^1f_i(v,w,c,t)dvdw=g_i(c)\quad \forall t>0\,\,\text{and} \,\,\forall i\in\mathbb{I}\,,
\]
where $g_i(c)$ is the fixed probability distribution of the neuronal connections in the $i$-th macro-area. Assuming that the pair $(\mathcal{X},\mathcal{V}_t,\mathcal{W}_t,\mathcal{C})$, $(\mathcal{X}^*,\mathcal{V}^*_t,\mathcal{W}^*_t,\mathcal{C}^*)$ consists in independent processes, i.e., their joint law reads ${f(x,v,w,c,t)f(x_*,v_*,w_*,c_*,t)}$, equation \eqref{limit_eq} can be re-written as 
\begin{equation}\label{limit_eq_inhomog}
{\allowdisplaybreaks 
\begin{split}
&\dfrac{d}{dt}\int\limits_{\mathbb{I}}\int\limits_0^1\int\limits_0^1\int\limits_0^{c_{\text{max}}}\phi(x,v,w,c)f(x,v,w,c,t)dcdwdvdx\\[0.15cm] 
&=\iint\limits_{\mathbb{I}}\iint\limits_0^1\iint\limits_0^1\iint\limits_0^{c_{\text{max}}}B(x,x_*)G(c,c_*)\left[\phi(x,v',w',c)- \phi(x,v,w,c)\right]\,\times\\[0.15cm]
&\hspace{3cm}f(x,v,w,c,t)f(x_*,v_*,w_*,c_*,t)dcdc_*dwdw_*dvdv_*dxdx_*\,.
\end{split}}
\end{equation}
Here, $(v',w'):=(v,w)+L(x,x_*,v,w,v_*,w_*)$ denotes the new state that a neuron with previous state $(x,v,w)$ acquires after an interaction with another neuron with state $(x_*,v_*,w_*)$, while $B(x,x_*)G(c,c_*)$ represents the collision kernel. Substituting \eqref{f_inhomo} into \eqref{limit_eq_inhomog}, fixing a certain index $i$ and assuming that $\phi(x,v,w,c)=\tilde{\phi}(x)\Phi(v,w,c)$, with $\tilde{\phi}(x)=1$ for $x=i$, while $\tilde{\phi}(x)=0$ for $x\ne i$, we obtain the weak form for the kinetic equation in the $i$-th macro-area
\begin{equation}\label{limit_eq_inhomog_i}
\begin{split}
&\dfrac{d}{dt}\int\limits_0^1\int\limits_0^1\int\limits_0^{c_{\text{max}}}\Phi(v,w,c)f_i(v,w,c,t)dcdwdv\\[0.15cm]
&=\sum_{j=1}^N\iint\limits_0^1\iint\limits_0^1\iint\limits_0^{c_{\text{max}}}B_{ij}G(c,c_*)[\Phi(v',w',c) -\Phi(v,w,c)]\,\times\\[0.15cm]
&\hspace{2.9cm}f_i(v,w,c,t)f_j(v_*,w_*,c_*,t)dcdc_*dwdw_*dvdv_*
\end{split}
\end{equation}
where $B_{ij}=B(x,x_*)$, with $x=i$ and $ x_*=j$, represents the probability of connection between the two macro-areas. Considering equation \eqref{limit_eq_inhomog_i}, we can analyze the evolution of the average membrane potential $V_i(t)$ and the average recovery variable $W_i(t)$ in the $i$-th macro-area, defined as
\begin{equation}
V_i(t):=\int\limits_0^{c_{\text{max}}}\!\int\limits_0^1\!\int\limits_0^1\!vf_i(v,w,c,t)dvdwdc\,,\,\, W_i(t):=\int\limits_0^{c_{\text{max}}}\!\int\limits_0^1\!\int\limits_0^1\!wf_i(v,w,c,t)dvdwdc\,,
\end{equation}
respectively. By choosing either $\Phi(v,w,c):=v$ or $\Phi(v,w,c):=w$ 
we obtain
\begin{equation}
\begin{split}
&\dfrac{dV_i(t)}{dt}=\sum_{j=1}^N\iint\limits_0^{c_{\text{max}}}\iint\limits_0^1\iint\limits_0^1B_{ij}G(c,c_*)(v'-v)f_i(v,w,c,t)\,\times\\[0.1cm]
&\hspace{4cm}f_j(v_*,w_*,c_*,t)dvdv_*dwdw_*dcdc_*\,,
\end{split}
\label{V_gen_inhomo}
\end{equation}
\begin{equation}
\begin{split}
&\dfrac{dW_i(t)}{dt}=\sum_{j=1}^N\iint\limits_0^{c_{\text{max}}}\iint\limits_0^1\iint\limits_0^1B_{ij}G(c,c_*)(w'-w)f_i(v,w,c,t)\,\times\\[0.1cm]
&\hspace{4cm}f_j(v_*,w_*,c_*,t)dvdv_*dwdw_*dcdc_*.
\end{split}
\label{W_gen_inhomo}
\end{equation}

\subsection{Microscopic rule for action potential propagation}\label{microRule}
Following \cite{conte2023kinetic}, for the description of the microscopic rule for the propagation of the action potential, we consider combined versions of the nondimensional Morris--Lecar neuron model \cite{morris1981voltage} and the FitzHugh--Nagumo model \cite{fitzhugh1961impulses}. Precisely, for $x=i$,
\begin{equation}\label{micro_rule}
L(x,x_*,v,w,v_*,w_*):=(\textbf{i}_{ext}^i+\gamma^i(\bar{v}-v)+(v_*-v)-w,v-aw)
\end{equation}
with $\textbf{i}^i_{ext}$ the external current sources stimulating the neuron located in the $i$-th macro-area, $\bar{v}$ the relaxation value for the membrane potential, and $\gamma^i>0$ and $a>0$ dimensionless parameters. We assume that $\gamma^i$ and $\textbf{i}^i_{ext}$ may differ among the macro-areas. Moreover, the term $(v_*-v)$ takes into account at the microscopic level the coupling between two different neurons. Assuming that $G(c,c_*)=1$, $\forall i\in\mathbb{I}$ from equations \eqref{V_gen_inhomo} and \eqref{W_gen_inhomo}, we obtain
\begin{equation}\label{macro_sis_inhom}
\begin{sistem}
\dfrac{dV_i(t)}{dt}\!=\!(\textbf{i}^i_{ext}+\gamma^i(\bar{v}-V_i(t))-W_i(t))\sum\limits_{j=1}^NB_{ij}+\sum\limits_{j=1}^NB_{ij}(V_j(t)-V_i(t))\\[0.4cm]
\dfrac{dW_i(t)}{dt}=(V_i(t)-aW_i(t))\sum\limits_{j=1}^NB_{ij}\,.
\end{sistem}
\end{equation}
Instead, if we consider the other choice $G(c,c_*)=c_*$, for equations \eqref{V_gen_inhomo} and \eqref{W_gen_inhomo} to be a closed system we need to define the mixed first-order moments
\begin{equation}
K_i^v(t):=\!\!\!\int\limits_0^{c_{\text{}max}}\!\!\int\limits_0^1\!\int\limits_0^1\!cvf_i(v,w,c,t)dvdwdc\,,\,\, K_i^w(t):=\!\!\!\int\limits_0^{c_{\text{}max}}\!\!\int\limits_0^1\!\int\limits_0^1\!cwf_i(v,w,c,t)dvdwdc.
\end{equation}
Thus, $\forall i\in\mathbb{I}$ the  macroscopic system for action potential dynamics reads
{\allowdisplaybreaks 
\begin{equation}\label{macro_sis_inhom_2}
\begin{sistem}
\dfrac{dV_i(t)}{dt}\!=\!(\textbf{i}^i_{ext}+\gamma^i(\bar{v}-V_i(t))-W_i(t))\sum\limits_{j=1}^NB_{ij}m_c^j+\sum\limits_{j=1}^NB_{ij}(K_j^v(t)-m_c^jV_i(t))\\[0.5cm]
\dfrac{dW_i(t)}{dt}=(V_i(t)-aW_i(t))\sum\limits_{j=1}^Nm_c^jB_{ij}\\[0.5cm]
\begin{split}
\dfrac{dK_i^v(t)}{dt}=&\,(m_c^i\textbf{i}^i_{ext}+\gamma^i(\bar{v}m_c^i-K_i^v(t))-K_i^w(t))\sum\limits_{j=1}^NB_{ij}m_c^j\\
&+\sum\limits_{j=1}^NB_{ij}(m_c^iK_j^v(t)-m_c^jK_i^v(t))
\end{split}\\
\dfrac{dK_i^w(t)}{dt}=(K_i^v(t)-aK_i^w(t))\sum\limits_{j=1}^Nm_c^jB_{ij}\,.\\
\end{sistem}
\end{equation}}
Here, $m_c^i$ indicates the mean of the distribution of the neuron connections in the $i$-th macro-areas, i.e.,
\begin{equation}\label{g_mean}
m_c^i:=\int_0^{c_\text{max}} cg_i(c)dc\,.
\end{equation}
Moreover, given the graph ${\mathcal{G}=(\mathbb{I},\mathbb{E})}$ with a set $\mathbb{I}= \{1,2,...,N\}$ of nodes (brain macro-area) and a set $\mathbb{E}$ of arcs $(i,j)$, we indicate with $B$ its adjacent matrix, which allows to identify the probability of connections $B_{ij}$ appearing in  the kinetic equations \eqref{limit_eq_inhomog_i}. Specifically, for $i=j$, we set $B_{ij}=1$, while for $i\ne j$ we set
\begin{equation}\label{laplacian}
B_{ij}=
\begin{sistem}
b_{ij}>0\quad\quad\,\,\,\text{if $(i,j)\in\mathbb{E}$} \\[0.2cm]
0\quad\,\quad\quad\quad\,\,\,\,\,\text{otherwise }\,.
\end{sistem}
\end{equation}
In the following, we set $B_i:=\sum_{j=1}^NB_{ij}$ and we consider the cases of either undirected graphs, where all edges are bidirectional, directed graphs, where the edges have a direction, weighted graphs, where each edge has an associated numerical value, or combinations of them.

\section{Equilibria of the kinetic and macroscopic systems}\label{Sec_Eq}
In this section, we study the equilibrium distributions of the deduced kinetic and macroscopic systems. Starting from system \eqref{macro_sis_inhom}, from \cite{conte2023kinetic} we know that, if ${\textbf{i}^i_{ext}=\textbf{i}_{ext}}$ and ${\gamma^i=\gamma}$, $\forall i\in\mathbb{I}$, there is a unique equilibrium configuration, given by the $2N$-dimensional vector $({\bf V}^*,{\bf W}^*)$ of components
\begin{equation}\label{eq_inhomo}
\begin{sistem}
V^*_i=\dfrac{\textbf{i}_{ext}+\gamma\bar{v}}{\gamma+\dfrac{1}{a}}\quad\,\,\,\forall i=1...N\\[.7cm]
W^*_i=\dfrac{1}{a}V^*_i\qquad\,\quad\,\forall i=1...N\,.
\end{sistem}
\end{equation}
Instead, in general, if $\textbf{i}^i_{ext}\ne \textbf{i}^j_{ext}$ and $\gamma^i\ne \gamma^j$, for $i\ne j$ and $i,j\in\mathbb{I}$, the above expression for the equilibrium configuration does not hold. However, it is possible to prove that the equilibrium of system \eqref{macro_sis_inhom} exists and it is unique.
\begin{proposition}
Let us consider system \eqref{macro_sis_inhom} with $\textbf{i}^i_{ext}\ne \textbf{i}^j_{ext}$ and $\gamma^i\ne \gamma^j$ for $i,j\in\mathbb{I}$, $i\ne j$. It admits a unique equilibrium configuration $({\bf V}^*,{\bf W}^*)$ of components
\begin{equation}\label{sol_inhomo}
\begin{sistem}
V^*_i=\tilde{V}_i\qquad\qquad\,\forall i=1...N\\[0.2cm]
W^*_i=\dfrac{1}{a}\tilde{V}_i\qquad\,\quad\,\forall i=1...N\,.
\end{sistem}
\end{equation}
\end{proposition}
\begin{proof}
Let us consider the vector $\tilde{\bf V}$ solution of the linear system 
$M{\bf \tilde{V}}= {\bf b}$
with load vector ${\bf b}$ of components ${b_i\!=\!B_i(\textbf{i}^i_{ext}+\gamma^i \bar{v})}$, for $i\in\mathbb{I}$, and associate matrix $M$ given by
\begin{equation}
M_{ij}=
\begin{sistem}
B_i\left(1+\gamma^i+\dfrac{1}{a}\right)-B_{ii} \quad\quad\,\,\ \text{if }\,i=j \\[0.4cm]
-B_{ij}\qquad\qquad\qquad\qquad\qquad\,\, \text{if }\,i\ne j \,.
\end{sistem}
\end{equation}
 $B$ is the adjacent matrix of the graph and $B_i=\sum_{j=1}^N B_{ij}$. We notice that $M$ is strictly diagonal dominant. In fact,
\begin{equation*}
\begin{split}
|M_{ii}|&=\left|B_i\left(1+\gamma^i+\dfrac{1}{a}\right)-B_{ii}\right|=\left|\sum_{\substack{j=1 \\ j\ne i}}^N B_{ij} +B_i\left(\gamma^i+\dfrac{1}{a}\right)\right|> \sum_{\substack{j=1 \\ j\ne i}}^N \left|M_{ij}\right| 
\end{split}
\end{equation*}
Thus, due to the properties of the strictly diagonal dominant matrices, $M$ is not singular. This ensures the existence of a unique solution $\tilde{{\bf V}}$ for the associated linear system $M{\bf \tilde{V}}= {\bf b}$. Therefore, system \eqref{macro_sis_inhom} has a unique equilibrium configuration  $({\bf V}^*,{\bf W}^*)$ given by \eqref{sol_inhomo}.
\end{proof}
A similar argument can be applied to system \eqref{macro_sis_inhom_2}. 
\begin{proposition}
Let us denote by ${\bf V}$, ${\bf W}$, ${\bf K}^v$, and ${\bf K}^w$ the vectors of components $V_i$, $W_i$, $K_i^v$ and $K^w_i$, respectively. Moreover, let us consider $\textbf{i}^i_{ext}\ne \textbf{i}^j_{ext}$ and $\gamma^i\ne \gamma^j$ for $i,j\in\mathbb{I}$, $i\ne j$. Thus, system \eqref{macro_sis_inhom_2} has a unique equilibrium configuration $({\bf V}^*,{\bf W}^*, {\bf K}^{v*},{\bf K}^{w*})$ of components
\begin{equation}\label{part_sol_inhomo_2}
\begin{sistem}
V^*_i=\tilde{V}_i\qquad\qquad\qquad\,\forall i=1...N\\[0.3cm]
W^*_i=\dfrac{1}{a}\tilde{V}_i\qquad\qquad\quad\,\forall i=1...N\,\\[0.4cm]
K^{v*}_i=\tilde{K}^v_i\qquad\qquad\quad\,\,\forall i=1...N\\[0.3cm]
K^{w*}_i=\dfrac{1}{a}\tilde{K}^{v}_i\quad\qquad\,\,\quad\,\forall i=1...N\,.
\end{sistem}
\end{equation}
\end{proposition}
\begin{proof}
Firstly, we notice that the equations for ${\bf K}^v$ and ${\bf K}^w$ in \eqref{macro_sis_inhom_2}  are independent of the ones for ${\bf V}$ and ${\bf W}$.  The related equilibrium distribution  $({\bf K}^{v*},{\bf K}^{w*})$ is such that ${K}^{v*}_i$ is the $i$-th component of the vector $\tilde{\bf K}^v$, solution of the linear system  $M^K{\bf \tilde{K}}^v= {\bf b}^K\,.$ This system has load vector ${\bf b}^K$ of components $b^K_i=m_c^i\tilde{B}_i(\textbf{i}^i_{ext}+\gamma^i \bar{v})$ and associated matrix $M^K$ of form 
\begin{equation}
M^K_{ij}=
\begin{sistem}
\tilde{B}_i\left(1+\gamma^i+\dfrac{1}{a}\right)-m_c^iB_{ii} \quad\quad\,\,\ \text{if }\,i=j \\[0.4cm]
-m_c^jB_{ij}\qquad\qquad\qquad\qquad\qquad\,\, \text{if }\,i\ne j 
\end{sistem}
\end{equation}
with $\tilde{B}_i:=\sum_{j=1}^N m_c^jB_{ij}$. This matrix is strictly diagonal dominant since
\begin{equation*}
\begin{split}
|M_{ii}|&=\left|\tilde{B}_i\left(1+\gamma^i+\dfrac{1}{a}\right)-m_c^iB_{ii} \right|=\left|\sum_{\substack{j=1 \\ j\ne i}}^N m_c^jB_{ij} +\tilde{B}_i\left(\gamma^i+\dfrac{1}{a}\right)\right|>\sum_{\substack{j=1 \\ j\ne i}}^N \left|M_{ij}\right| 
\end{split}
\end{equation*}
 This ensures the existence of a unique equilibrium $({\bf{K}}^{v*},{\bf K}^{w*})$ for ${\bf K}^v$ and ${\bf K}^w$ in \eqref{part_sol_inhomo_2}. Moreover, having this solution, it is possible to deduce the unique equilibrium configuration $({\bf V^*},{\bf W}^*)$. It has the expression in \eqref{sol_inhomo} with ${\bf\tilde{V}}$ solution of the linear system  $M^V{\bf \tilde{V}}= {\bf b}^V\,,$
 with load vector ${\bf b}^V$ of components ${b_i=\tilde{B}_i\left(\textbf{i}^i_{ext}+\gamma^i \bar{v}\right)+\sum_{j=1}^NB_{ij}K^v_j}$ and non singular diagonal matrix $M^V$ such that $M_{ii}^V=\tilde{B}_i\left(1+\gamma^i+\frac{1}{a}\right)$.
 \end{proof}

The knowledge of the equilibrium configurations for the macroscopic systems describing the evolution of the average membrane potential and recovery variable allows  to easily deduce the equilibrium distribution for their kinetic counterpart. 
\begin{proposition}\label{eq_kinetic}
Let us consider the kinetic model in the weak form given by equation \eqref{limit_eq_inhomog_i}. Thus, $\forall i\in\mathbb{I}$ and $\forall (v,w) \in [0,1]^2\,,c \in [0, c_{\text{max}}]\,, t\ge 0$, the distribution function defined by 
\begin{equation}\label{f_eq_kin}
f_i^0(v,w,c,t)=\delta(v-V_i^*)\delta(w-W_i^*)g_i(c)
\end{equation}
 is equilibrium distribution of equation \eqref{limit_eq_inhomog_i}. 
\end{proposition}
\begin{proof}
For a fixed $i\in\mathbb{I}$, let us consider a generic term in the sum on the right-hand side of equation \eqref{limit_eq_inhomog_i}. Dropping the dependency on $(v,w,c)$ in the distribution functions, if we evaluate this term in $f_i(v,w,c,t)=f_i^0(v,w,c,t)$ and ${f_j(v_*,w_*,c_*,t)=f_j^0(v_*,w_*,c_*,t)}$ we have 
\begin{equation}
\begin{split}
&\iint\limits_0^1\iint\limits_0^1\iint\limits_0^{c_{\text{max}}}B_{ij}G(c,c_*)[\Phi(v',w',c) -\Phi(v,w,c)]f_i^0f_j^0dcdc_*dwdw_*dvdv_*\\[0.4cm]
&=\iint\limits_0^1\iint\limits_0^1\iint\limits_0^{c_{\text{max}}}B_{ij}G(c,c_*)[\Phi(v',w',c) -\Phi(v,w,c)]\times\\[0.2cm]
&\hspace{1cm}\delta(v-V_i^*)\delta(w-W_i^*)g_i(c)\delta(v_*-V_j^*)\delta(w_*-W_j^*)g_j(c_*)dcdc_*dwdw_*dvdv_*\\[0.4cm]
&=\iint\limits_0^{c_{\text{max}}}\left[\Phi\left(V_i^*+\textbf{i}^i_{ext}+\gamma^i(\bar{v}-V_i^*)-W_i^*,\frac{1}{a}V_i^*,c\right) -\Phi\left(V_i^*,W_i^*,c\right)\right]\times\\[0.2cm]
&\hspace{1cm}B_{ij}G(c,c_*)g_i(c)g_j(c_*)dcdc_*\\[0.4cm]
&=\iint\limits_0^{c_{\text{max}}}B_{ij}G(c,c_*)\left[\Phi\left(V_i^*,W_i^*,c\right) -\Phi\left(V_i^*,W_i^*,c\right)\right]g_i(c)g_j(c_*)dcdc_*=0\,.
\end{split}
\end{equation}
Thus, the distribution function $f_i^0(v,w,c,t)$ in \eqref{f_eq_kin} is equilibrium distribution of \eqref{limit_eq_inhomog_i}, $\forall i\in\mathbb{I}$.
\end{proof}
As a consequence of Proposition \ref{eq_kinetic}, we observe that the distribution
\begin{equation}\label{equil_f_inhomo}
f^*=f^*(x,v,w,c,t):=\sum_{i=1}^Nf_i^0(v,w,c,t)\delta(x-i)
\end{equation}
is equilibrium distribution for equation \eqref{limit_eq_inhomog}.

\section{Numerical simulations}
\label{simulation}
In this section, we present some numerical tests obtained by numerical integration of the macroscopic systems \eqref{macro_sis_inhom} and \eqref{macro_sis_inhom_2}. We recall that system \eqref{macro_sis_inhom} is derived from  \eqref{V_gen_inhomo}-\eqref{W_gen_inhomo} with the microscopic rule  \eqref{micro_rule} and assuming $G(c,c_*)=1$, while system \eqref{macro_sis_inhom_2} is derived from  \eqref{V_gen_inhomo}-\eqref{W_gen_inhomo} with the microscopic rule  \eqref{micro_rule} and assuming $G(c,c_*)=c_*$. For the numerical simulations, we use a fourth order Runge--Kutta method. We set $\bar{v}=1$, $a=0.6$, and we vary $\gamma^i$ and $\textbf{i}^i_{ext}$ as indicated in each numerical tests. The initial conditions are set to $V(0)=0$ and $W(0)=0$ for system \eqref{macro_sis_inhom}, while for system \eqref{macro_sis_inhom_2} we consider the additional initial conditions $K^v(0)=0$ and $K^w(0)=0$ .

To show the capability of the proposed approach to capture a large variety of scenarios, we present three main numerical tests.
\begin{itemize}
\item[{\bf Test 1:}] in Section \ref{OrWe}, we consider system \eqref{macro_sis_inhom} in the case of directed and weighted graphs. This combined choice of directed and weighted graphs is motivated by the fact that the membrane potential dynamics are influenced by a combination of excitatory/inhibitory connections, which are unidirectional, and the physical distance between the brain macro-areas, taken into account  by weighting the graph edges.
\item[{\bf Test 2:}] in Section \ref{GammaCurr}, we analyze the effect on system \eqref{macro_sis_inhom} of heterogeneities in the system coefficients $\gamma^i$ and $\textbf{i}^i_{ext}$, for $i\in\mathbb{I}$. This allows us to observe the emergence of asynchronous equilibrium manifolds in the system evolution. 
\item[{\bf Test 3:}] in Section \ref{Asim_conn}, we consider system \eqref{macro_sis_inhom_2} assuming different characteristics for the distribution of the connections $g_i(c)$ in each node of the graph, showing the impact on the spiking dynamics. 
\end{itemize}

\subsection{Test 1: directed and weighted graph}\label{OrWe}
We analyze here two examples ({\bf A} and {\bf B}) of directed and weighted graphs, consisting of $N=5$ macro-areas differently connected, and we consider the dynamics given by system \eqref{macro_sis_inhom}. The configurations of nodes and edges are described through the adjacency matrix $B_{ij}$, while we set $\textbf{i}^i_{ext}=0.5$, $\gamma^i=0.7$, $\forall i=1,\dots,5$. 

Firstly, as {\bf Example A}, we consider the graph described by the following adjacency matrix 
\begin{equation}\label{B_exA} 
B_{\bf A}=
\begin{pmatrix}
1 &\,\,0.25&\,\, 0&\,\, 0&\,\,\, 0\\
0&\,\,1 &\,\,0 &\,\,0.5 &\,\,\,1\\
0&\,\, 0&\,\, 1&\,\, 0&\,\, \,0\\
0 &\,\,0 &\,\,0.25 &\,\,1 &\,\,\,0\\
0 &\,\,0 &\,\,1 &\,\,0 &\,\,\,1
\end{pmatrix}\,.
\end{equation}
The corresponding graph is schematized in the left plot of Figure \ref{Graph_exAB}. Instead, as {\bf Example B}, we consider a ring graph described by the following adjacency matrix 
\begin{equation}\label{B_exB} 
B_{\bf B}=
\begin{pmatrix}
1 &\,\,\,\,1&\,\, 0&\,\, 0&\,\, 1\\
0&\,\,\,\,1 &\,\,0.5 &\,\,0 &\,\,0\\
0&\,\, \,\,0&\,\, 1&\,\, 0.25&\,\, 0\\
0 &\,\,\,\,0 &\,\,0&\,\,1 &\,\,0.5\\
0 &\,\,\,\,0 &\,\,0&\,\,0 &\,\,1
\end{pmatrix}\,.
\end{equation}
and the corresponding graph is schematized in the right plot of Figure \ref{Graph_exAB}. 
\begin{figure}[!h]
     \centering
  \includegraphics[width=.85\textwidth]{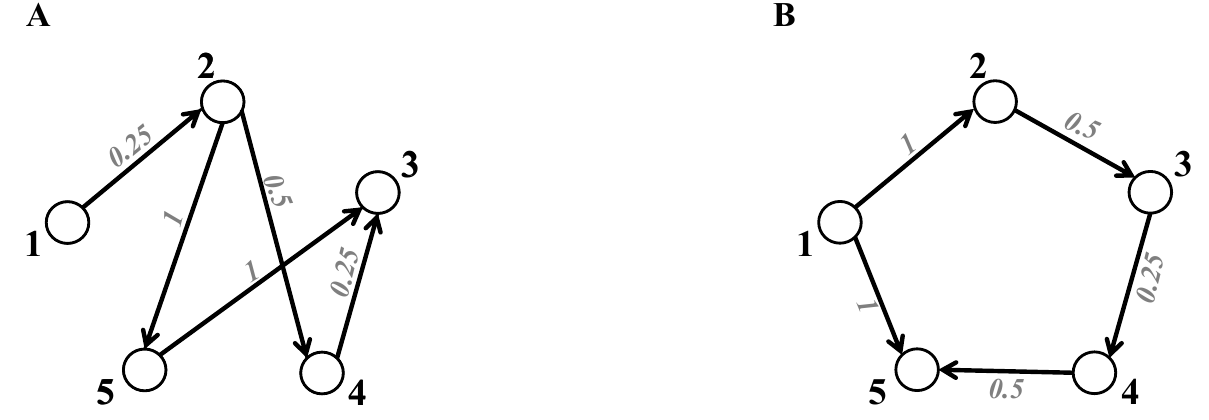}
\caption{{\bf Test 1.} Schematic representation of the graphs described by the adjacent matrices \eqref{B_exA} (Example {\bf A} - left plot) and  \eqref{B_exB} (Example {\bf B} - right plot), respectively.}
    \label{Graph_exAB}
\end{figure}
The  dynamics of the potentials $V_i(t)$ emerging from system \eqref{macro_sis_inhom} with these two different network topologies are shown in Figure \ref{Dynam_exAB}.
\begin{figure}[!h]
     \centering
  \includegraphics[width=.87\textwidth]{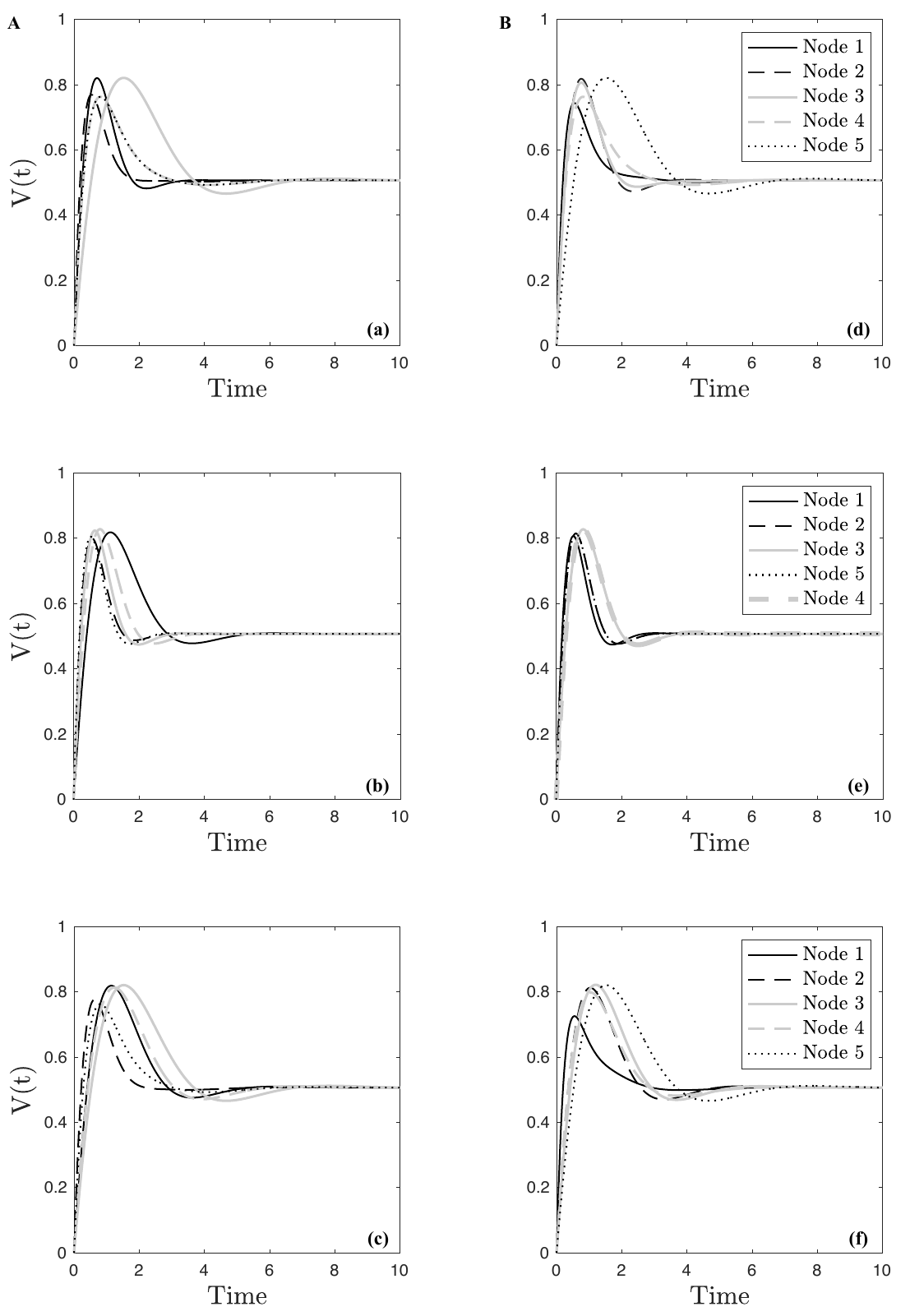}
\caption{{\bf Test 1.} Dynamics of the potentials $V_i(t)$, for $i=1,\dots,5$, emerging from system \eqref{macro_sis_inhom} on the networks schematized in Figure \ref{Graph_exAB}. Left column (panels a-c) refers to Example {\bf A}, while right column (panels d-f) to Example {\bf B}. The rows show the cases of directed-non weighted (top), undirected-weighted (middle), and directed-weighted (bottom) graphs, respectively.}
    \label{Dynam_exAB}
\end{figure}
\noindent As a matter of comparison, we include in the first two rows some aspects identified in \cite{conte2023kinetic}. Precisely, for a directed-non weighted graph (panels a and d) the timing for neuron spike is related to the net outflow of each node, identified with 
\begin{equation}\label{Oquan}
{O}_i=\sum_{j=1}^N (B_{ij}-B_{ji})\,.
\end{equation}
The higher the value of $O_i$ is, the earlier the spike occurs. Then, if the graph is undirected and weighted (panels b and e) , the timing for neuron spike
is related to the quantity
\begin{equation}\label{Tquan}
T_i=\sum_{j=1}^NB_{ij}
\end{equation}
which is the sum of the weights on the edges related to the $i$-th node. The greater $T_i$ is, the faster the $i$-th action potential spikes. When a combination of direction and weights for each edge (panels c and f)  is considered, the dynamics of the $i$-th node are characterized by the same quantity $T_i$ defined in \eqref{Tquan}. However, in this case, it represents the sum of the weights related only to the outgoing edges. This means that nodes which are closer to others and with more outgoing edges are more likely to show an earlier action potential spike than farther nodes or nodes with more (or only) incoming edges. Moreover, we can observe that, for the nodes with the same value of $O_i$ (for the directed-non weighted graphs) or $T_i$ (for the undirected/directed and weighted graphs), the differences in their action potential dynamics are determined by the characteristics of the node to which they are connected. For completeness, we collect all the values of $O_i$ and $T_i$ for each node in the case of directed-non weighted (a and d),  undirected-weighted (b and e), and directed-weighted (c and f) graphs in Tables \ref{TabTest1_A} (for Example {\bf A}) and \ref{TabTest1_B} (for Example {\bf B}), respectively. In these tables, for each case, we also indicate the timing of neuron spike $t_S$ that represents the time at which the potential $V_i(t)$ reaches its maximum. For instance, referring to Example {\bf A}, in Figure \ref{Dynam_exAB}-a node 3 has $O_3=-2$, the smallest value among all the nodes, and, in fact, it is the last one to spike, with a timing of neuron spike $t_S=1.51$, as indicated in Table \ref{TabTest1_A}. Moreover, we observe that while nodes 4 and 5 have identical values for $O_i$ and are connected to the same nodes (2 and 3), leading to synchronized neuron spike timings, nodes 1 and 2, despite having the same $O_i$ value, exhibit different spike timings $t_S$. This difference arises from their distinct connection patterns: node 1 is connected solely to node 2, while node 2 also has additional connections to nodes 4 and 5, resulting in faster spike timing. For the undirected-weighted case in Figure \ref{Dynam_exAB}-b, node 3 has $T_3=1.25$, which is the intermediate value between the fastest nodes (2 and 5) and the slowest ones (1 and 4). When we combine directions and weights in Figure \ref{Dynam_exAB}-c, then the quantity $T_3$ becomes $T_3=0$, because node 3 does not have any outgoing edge. As it is the only one without any outgoing edge, it results to be the one with the slowest dynamics. Similar observations can be done on Example {\bf B}, looking, for instance, at the dynamics of node 5, which has $O_5=-2$ in Figure \ref{Dynam_exAB}-d (the smaller $O_i$ value) and $T_5=1.5$ in Figure \ref{Dynam_exAB}-e (the second bigger value, after $T_1=2$). However, when in Figure \ref{Dynam_exAB}-f a combined directed-weighted graph is considered, then $T_5=0$, therefore, the dynamics of node 5 is the slowest one. Similar arguments as above apply to Figure \ref{Dynam_exAB}-d to explain why the nodes 2, 3, and 4 share the same value of $O_i$ but exhibit different spiking times. In contrast, in Figure \ref{Dynam_exAB}-e, nodes 3 and 4 have the same value of $T_i$ and are connected to nodes 2 and 5, which also share the same $T_i$ values, leading to synchronization between nodes 3 and 4. Therefore, these observations sustain our hypothesis that directed connections among nodes support action potential activation and its faster propagation in the network, as well as their spatial nearness. 

 \begin{table}[h!]
\begin{center}
   \begin{tabular}{c|c|c|c|c|c|c}
   \toprule  
   \rule{0pt}{2.5ex}  \multirow{2}{*}{Nodes} \,\,& \multicolumn{2}{|c|}{{\bf (a)}}& \multicolumn{2}{|c|}{{\bf (b)}} & \multicolumn{2}{|c}{{\bf (c)}} \\\cline{2-7}
   \rule{0pt}{2.5ex} &\,\,\,\,$O_i$\,\, \,\,& \,\,\,\,$t_S$ \,\,\,\,& \,\,\,\,$T_i$\,\,\,\,&\,\, \,\,$t_S$\,\,\,\,&\,\, \,\,$T_i$\,\,\,\,&\,\,\,\,$t_S$\,\, \,\, \\
  \midrule
  \rule{0pt}{3ex} 1 & 1 &0.71 &0.25 &1.13 &0.25 &1.15 \\[0.5ex]\hline
   \rule{0pt}{3ex} 2 & 1 &0.55 & 1.75& 0.58& 1.5&0.61 \\[0.5ex]\hline
   \rule{0pt}{3ex} 3 & -2&1.51 & 1.25& 0.64&0 & 1.53\\[0.5ex]\hline
   \rule{0pt}{3ex} 4 & 0 & 0.81& 0.75& 0.81& 0.25&1.24 \\[0.5ex]\hline
   \rule{0pt}{3ex} 5 & 0 & 0.81&2 &0.52 & 1&0.83 \\[0.5ex]
\bottomrule
    \end{tabular}
\end{center}
\caption{{\bf Test 1-A.} Values for the quantities $O_i$, defined in \eqref{Oquan}, $T_i $, defined in \eqref{Tquan}, and $t_S$, which describes the timing for neuron spike, in the simulations shown in Figure \ref{Dynam_exAB}-A.}
 \label{TabTest1_A}
\end{table}

 \begin{table}[h!]
\begin{center}
   \begin{tabular}{c|c|c|c|c|c|c}
   \toprule  
   \rule{0pt}{2.5ex}  \multirow{2}{*}{Nodes} \,\,& \multicolumn{2}{|c|}{{\bf (d)}}& \multicolumn{2}{|c|}{{\bf (e)}} & \multicolumn{2}{|c}{{\bf (f)}} \\\cline{2-7}
   \rule{0pt}{2.5ex} &\,\,\,\,$O_i$\,\, \,\,& \,\,\,\,$t_S$ \,\,\,\,& \,\,\,\,$T_i$\,\,\,\,&\,\, \,\,$t_S$\,\,\,\,&\,\, \,\,$T_i$\,\,\,\,&\,\,\,\,$t_S$\,\, \,\, \\
  \midrule
   \rule{0pt}{3ex} 1 & 2 &0.51 &2 & 0.52 &2 &0.54 \\[0.5ex]\hline
   \rule{0pt}{3ex} 2 & 0 &0.76 & 1.5& 0.6& 0.5& 1.05\\[0.5ex]\hline
   \rule{0pt}{3ex} 3 & 0&0.77 & 0.75& 0.85&0.25&1.18 \\[0.5ex]\hline
   \rule{0pt}{3ex} 4 & 0 & 0.8& 0.75&0.85 & 0.5& 1.06\\[0.5ex]\hline
   \rule{0pt}{3ex} 5 & -2 &1.55 &1.5 &0.6 & 0&1.52 \\[0.5ex]
\bottomrule
    \end{tabular}
\end{center}
\caption{{\bf Test 1-B.} Values for the quantities $O_i$, defined in \eqref{Oquan}, $T_i $, defined in \eqref{Tquan}, and $t_S$, which describes the timing for neuron spike, in the simulations shown in Figure \ref{Dynam_exAB}-B.}
 \label{TabTest1_B}
\end{table}

\subsection{Test 2: heterogenous parameter network}\label{GammaCurr}
We analyze here the effect of possible heterogeneities in the nodes (and, thus, in the brain macro-areas) characteristics which can influence the dynamics of system  \eqref{macro_sis_inhom}. We consider the simplest situation of undirected and non-weighted ring graph of $N=5$ nodes. If $\textbf{i}^i_{ext}=0.5$, $\gamma^i=0.7$, $\forall i=1,\dots,5$, then the emerging dynamics are characterized by a full synchronization of the action potentials, since all the nodes are identical and it is not possible to distinguish their dynamics, as shown in Figure \ref{GammaIcurr_equal} (right plot).

\begin{figure}[!h]
     \centering
  \includegraphics[width=.82\textwidth]{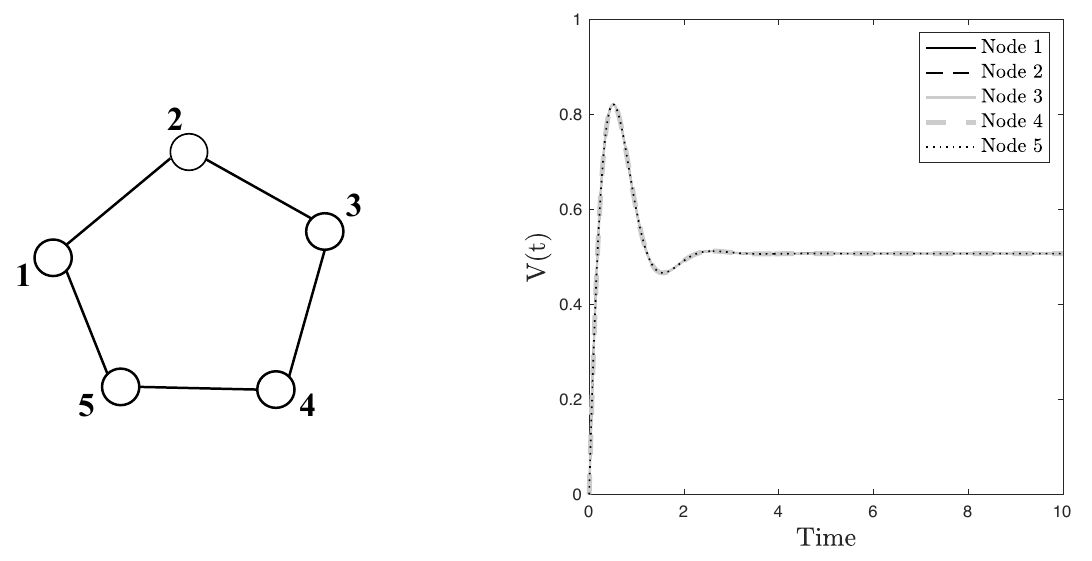}
\caption{{\bf Test 2.} Schematic representation of an undirected and non-weighted ring graph (left plot). Dynamics of the potentials $V_i(t)$, for $i=1,\dots,5$ emerging from system \eqref{macro_sis_inhom} (right plot).}
    \label{GammaIcurr_equal}
\end{figure}

\noindent Instead, if we assume that the values of the external currents $\textbf{i}^i_{ext}$ and the relaxation parameters $\gamma^i$ may change across the nodes, different dynamics emerge. In Figure \ref{GammaIcur} we show the system evolution for $\gamma^i=0.1+0.125\cdot(i-1)$  and $\textbf{i}^i_{ext}=0.5$, $\forall i=1,\dots,5$ (left plot), or ${\textbf{i}^i_{ext}=0.1+0.1\cdot(i-1)}$ and  $\gamma^i=0.7$, $\forall i=1,\dots,5$ (right plot). 
 \begin{figure}[!h]
     \centering
  \includegraphics[width=.95\textwidth]{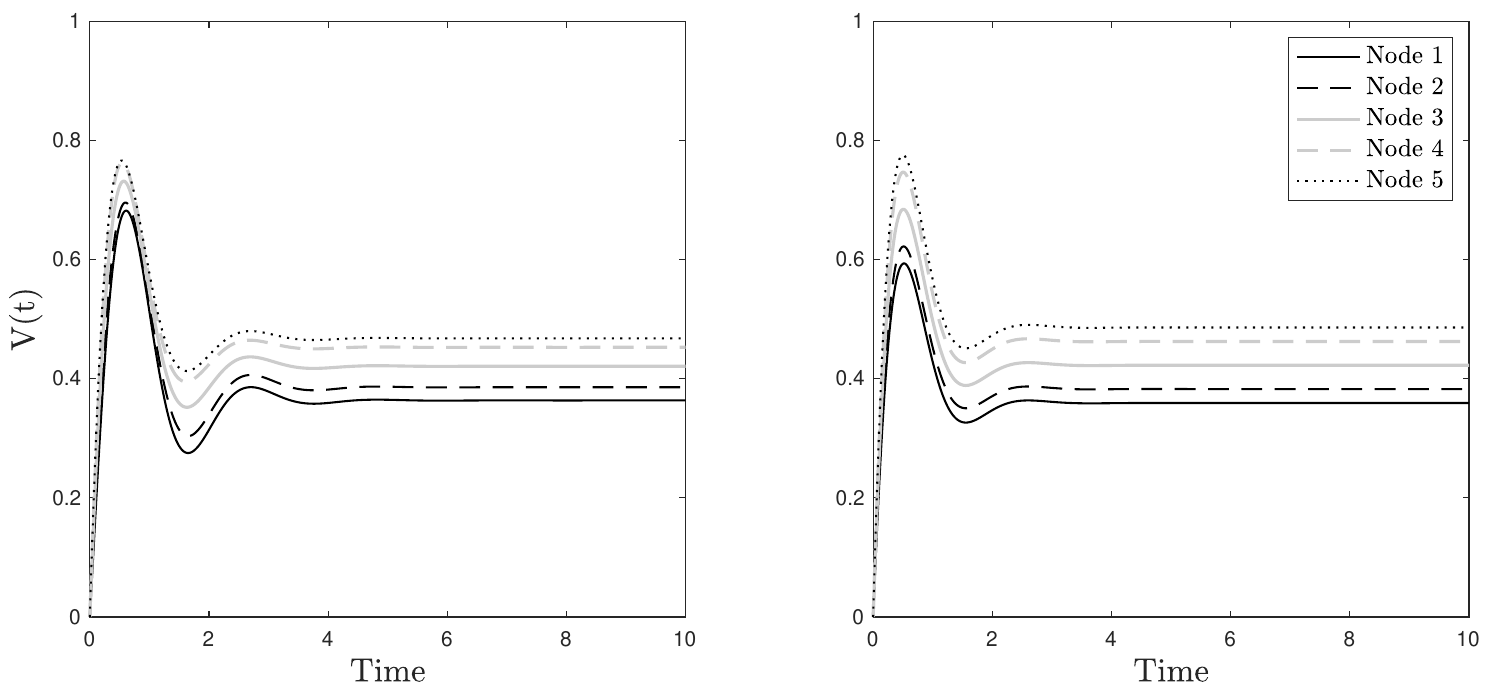}
\caption{{\bf Test 2.} Dynamics of the potentials $V_i(t)$, for $i=1,\dots,5$, emerging from system \eqref{macro_sis_inhom} on the network schematized in the left plot of Figure \ref{GammaIcurr_equal} when different values of $\gamma^i$ (left plot) and $\textbf{i}^i_{ext}$ (right plot) are assigned to the different nodes.}
    \label{GammaIcur}
\end{figure}
From Figure \ref{GammaIcur} we observe that varying the values of $\gamma^i$ across the nodes affects the resting value on which the potentials tend to relax, as well as the amplitude of the undershoot. Even if the timing of the spike does not change across the nodes and with respect to the dynamics shown in Figure \ref{GammaIcurr_equal}, we notice that the greater $\gamma^i$, the bigger the value of the resting potential on which each node stabilizes and the smaller the amplitude of the corresponding undershoot. This double effect is due to the fact that there is a nonlinear dependency on $\gamma^i$ of the equilibrium distribution of the potentials. Instead, varying $\textbf{i}^i_{ext}$ across the nodes does not affect the shape of the spike, but simply causes a shift in the action potential profiles, which maintains the same amplitude for both the peak and the undershoot and linearly changes the value of the resting state.

\subsection{Test 3: heterogenous distribution of the connections}\label{Asim_conn}
The last test analyzes the dynamics of system \eqref{macro_sis_inhom_2} in the special case of the ring graph illustrated in Figure \ref{GammaIcurr_equal} (left plot). In particular, we assume that each node has a different distribution function of the connections $g_i(c)$, characterizing it by its mean value $m_c^i$ defined in \eqref{g_mean}. In Figure \ref{HeterConnection}, we show the evolution of system  \eqref{macro_sis_inhom_2} for $m_c^i=0.5$, $\forall i=1,\dots,5$ (left plot) or for ${m_c^i=1-0.225\cdot(i-1)}$, $\forall i=1,\dots,5$ (right plot).
\begin{figure}[!h]
     \centering
  \includegraphics[width=.95\textwidth]{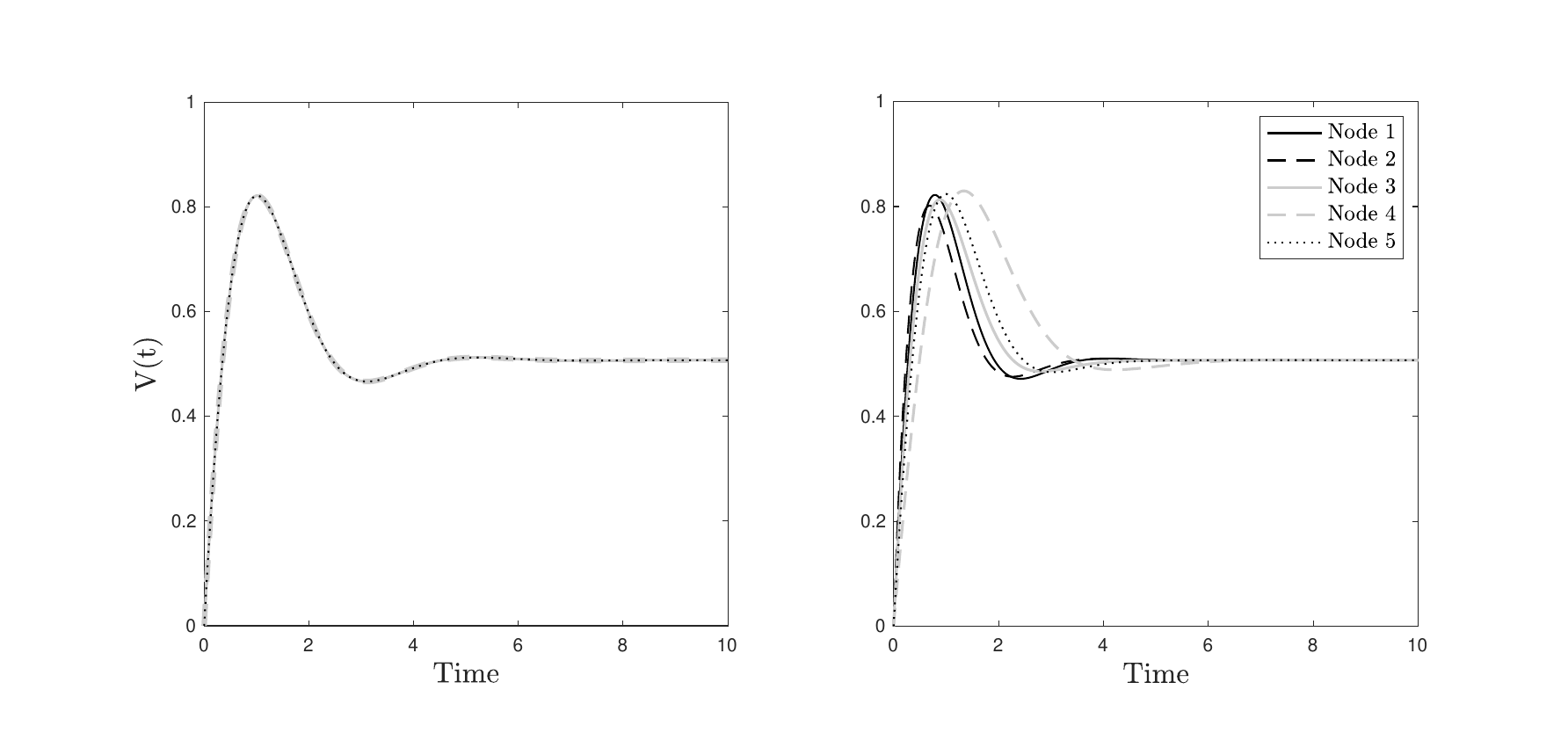}
\caption{{\bf Test 3.} Dynamics of the potentials $V_i(t)$, for $i=1,\dots,5$, emerging from system \eqref{macro_sis_inhom_2} on the networks schematized in the left plot of Figure \ref{GammaIcurr_equal} for the same value (left plot) or different values (right plot) of $m_c^i$ across the nodes.}
    \label{HeterConnection}
\end{figure}
The difference between systems \eqref{macro_sis_inhom_2} and \eqref{macro_sis_inhom} relies on the choice of $G(c,c_*)$ that, consequently, determines different macroscopic equations for the potentials. Assuming the same mean for the connection distribution within each node and comparing the right plot of Figure \ref{GammaIcurr_equal} with the left plot of Figure \ref{HeterConnection} we notice that the asymmetric choice of the kernel $G(c,c_*)$ determines a slowdown of the spiking dynamics, with spikes of similar amplitude but longer duration. However, it does not de-synchronize the network. If, instead, we consider different mean values of the distribution $g_i(c)$, then a different scenario emerges. In this case, the spiking dynamics depend on the quantity
\[
m_i=\sum_{j=1}^N B_{ij}m_c^j
\]
that, for each $i$-th node, takes into account the mean of the distribution function of the $i$-th node and the ones connected to it. The greater the value of $m_i$ the faster the emergence of the peak. The amplitude, instead, seems to be not affected. In the example of Figure \ref{HeterConnection} (right plot), for instance, node 4 (the slowest) has $m_4=0.975$, while node 2 (the fastest) has $m_2=2.325$. Therefore, even when a discrete ring network with no differences in the node parameters or the edge orientations and weights is considered, the asymmetric assumption on $G(c,c_*)$ allows the emergence of  heterogenous scenarios in the membrane potential evolution.

\section{Discussion}\label{conclusion}

In this note, we studied the dynamics of the action potential propagation with a combination of statistical and discrete descriptions of the brain connections, accounting for heterogeneity in the brain regions. We started from the spatially inhomogeneous framework established in \cite{conte2023kinetic} where the coupling between a local and statistical description of the interactions between neurons in a single brain region with a non-local and discrete description of the interactions between connected brain areas was considered. In particular, we extended the study in order to account for different types of heterogeneities in the nodes representing the brain macro-areas.

Under the assumption $G(c,c_*)=1$, we considered a microscopic rule that accounts for heterogeneities across the nodes in either the external source current or the relaxation parameter. We derived the corresponding macroscopic system for the average quantities and we analyzed its equilibrium manifold. In particular, it is possible to analytically show that there exists a unique equilibrium configuration for the system, independently of the network topology. Instead, assuming $G(c,c_*)=c_*$, we showed that the derivation of a closed macroscopic system requires to consider additional equations for the mixed first-order moments. We analyzed the corresponding macroscopic system, proving the existence of a unique equilibrium for a general graph. We also proved a general result on the underlying kinetic model, namely we showed that it admits at least one equilibrium distribution.

With the numerical tests proposed in the last section, we reinforced our observations about the influence of the specific brain region connections on the overall dynamics, showing how combinations of directed and weighted edges in the graph may affect the propagation of the action potential. Moreover, we analyzed the impact of the heterogeneities in the nodes' parameters on the dynamics, noticing a de-synchronization of the resting state of the action potential, as well as changes in the shape of the spike. Finally, we showed the effect of choosing an asymmetric interaction kernel $G(c,c_*)$, connecting the timing of the spike to the characteristics of the distribution of the connections in each node and its neighbors. 

Future plans will involve a more realistic description of the neural network, as the simple five-node network has been chosen only for an illustrative purpose. Moreover, we plan to use this modeling framework for other applications, specifically involving the description of disease progression into the brain \cite{bertsch2017JPA,bertsch2017alzheimer}.

\section*{Acknowledgement}
The authors  would like to thank Prof. Andrea Tosin for helpful discussions. The authors are members of INdAM-GNFM. The authors MB and MG acknowledge the support from COST Action CA18232 MAT-DYN-NET. The work has been performed in the frame of the project PRIN 2022 PNRR {\it ''Mathematical Modelling for a Sustainable Circular Economy in Ecosystems''} (project code P2022PSMT7) funded by the European Union - Next Generation EU, PNRR-M4C2-I 1.1 CUP:  D53D23018960001, and by MUR-Italian Ministry of Universities and Research. MC has been partially supported by the State Research Agency of the Spanish Ministry of Science and FEDER-EU, project PID2022-137228OB-I00 (MICIU/AEI /10.13039/501100011033); by Modeling Nature Research Unit, Grant QUAL21-011 funded by Consejería de Universidad, Investigaci\'on e Innovaci\'on (Junta de Andalucía)
The authors MB and MG also thank the support of the University of Parma through the action Bando di Ateneo 2022 per la ricerca, co-funded by MUR-Italian Ministry of Universities and Research - D.M. 737/2021 - PNR - PNRR - Next Generation EU (project {\it ''Collective and self-organised dynamics: kinetic and network approaches''}), and of the PRIN 2020 project {\it ''Integrated Mathematical Approaches to Socio–Epidemiological Dynamics''} (Prin 2020JLWP23, CUP: E15F21005420006).

%
%
                                    
\end{document}